\documentclass[a4paper,USenglish,cleveref, autoref, thm-restate, pdfa]{lipics-v2021-modified}

\title{Pathwidth of 2-Layer $k$-Planar Graphs}

\author{Yuto Okada}
{Nagoya University, Japan \and \url{https://yutookada.com/en/}}
{research@yutookada.com}
{https://orcid.org/0000-0002-1156-0383}
{}

\authorrunning{Y. Okada}
\Copyright{Yuto Okada}

\funding{Supported by JST SPRING, Grant Number JPMJSP2125.}

\acknowledgements{I would like to thank Nicol\'as Honorato-Droguett for proofreading the manuscript. I would also like to thank the anonymous reviewers for very helpful suggestions.}

\nolinenumbers
\hideLIPIcs

\newcommand{\pw}{\mathrm{pw}}

\newcommand{\ns}{\mathrm{ns}}

\begin{document}

\maketitle

\begin{abstract}
    A bipartite graph $G = (X \cup Y, E)$ is a 2-layer $k$-planar graph if it admits a drawing on the plane such that the vertices in $X$ and $Y$ are placed on two parallel lines respectively, edges are drawn as straight-line segments, and every edge involves at most $k$ crossings.
    Angelini, Da Lozzo, F\"orster, and Schneck~[GD 2020; Comput. J., 2024] showed that every 2-layer $k$-planar graph has pathwidth at most $k + 1$.
    In this paper, we show that this bound is sharp by giving a 2-layer $k$-planar graph with pathwidth $k + 1$ for every $k \geq 0$.
    This improves their lower bound of $(k+3)/2$.
\end{abstract}

\section{Introduction}

A \emph{2-layer drawing} of a bipartite graph $G$ with bipartition $(X, Y)$ is a drawing on the plane obtained by placing the vertices in $X$ on a line (layer), placing the vertices in $Y$ on another parallel line (layer), and drawing the edges as straight-line segments.
This drawing style is not only a natural model for drawing bipartite graphs, but also has an application to \emph{layered drawing}, which is similarly defined, but may have many layers:
the Sugiyama method, a method for producing a layered drawing of a directed graph introduced by Sugiyama, Tagawa, and Toda~\cite{DBLP:journals/tsmc/SugiyamaTT81} employs the crossing minimization problem on the 2-layer model as a subroutine.

Due to their importance, many graph classes admitting good 2-layer (or $h$-layer) drawings have been introduced, and their recognition algorithms have been studied in the literature.
The crossing minimization problems for $2$-layer and $h$-layer drawings are both NP-complete~\cite{doi:10.1137/0604033, DBLP:journals/siamcomp/HeathR92}.
However, they admit FPT algorithms with respect to $h + c$, where $c$ is the minimum number of edge crossings~\cite{DBLP:journals/algorithmica/DujmovicFKLMNRRWW08}.
Angelini, Da Lozzo, F\"orster, and Schneck~\cite{TheUpperBoundGD2020,TheUpperBoundComputJ2024} initiated the study of \emph{2-layer $k$-planar graphs}, the graphs that admit a 2-layer drawing such that every edge involves at most $k$ crossings.
Kobayashi, Okada, and Wolff~\cite{kobayashi_et_al:LIPIcs.SoCG.2025.65} gave an XP algorithm for recognizing 2-layer $k$-planar graphs with respect to $k$, which yields a polynomial-time algorithm for every fixed $k$.
They also showed that the recognition problem is XNLP-hard and hence admits no FPT algorithm under a plausible assumption.
\emph{Fan-planar drawings} with $h$ layers have also been studied~\cite{DBLP:conf/mfcs/BiedlC0MNR20}.
In a fan-planar drawing, an edge can cross other edges any number of times while the edges crossed by a single edge have a common endpoint.
For recognizing 2-layer fan-planar graphs, linear-time algorithms are known for trees~\cite{DBLP:conf/mfcs/BiedlC0MNR20} and biconnected graphs~\cite{DBLP:journals/jgaa/BinucciCDGKKMT17}.
For general graphs, Kobayashi and Okada~\cite{2508.17349} recently gave a polynomial-time algorithm, by incorporating fan-planarity into the algorithm of~\cite{kobayashi_et_al:LIPIcs.SoCG.2025.65} for recognizing 2-layer $k$-planar graphs.

As layered drawings have linear shapes, those classes often have bounded pathwidth.
The class of bipartite graphs that admit a crossing-free 2-layer drawing is equivalent to the class of caterpillars, which have pathwidth at most 1.
More generally, the graphs admitting an $h$-layer drawing with $k$ edge crossings have pathwidth at most $h+2k-1$~\cite{DBLP:journals/algorithmica/DujmovicFKLMNRRWW08}.
Angelini, Da Lozzo, F\"orster, and Schneck~\cite{TheUpperBoundGD2020,TheUpperBoundComputJ2024} showed that 2-layer $k$-planar graphs have pathwidth at most $k+1$, for which they also gave a lower bound of $(k+3)/2$.
The authors in~\cite{DBLP:conf/mfcs/BiedlC0MNR20} showed that $h$-layer fan-planar graphs have pathwidth at most $2h-2$.
Recently, Wood~\cite{DBLP:journals/jgaa/Wood23} characterized the pathwidth-boundedness of bipartite graphs by the existence of a certain 2-layer drawing.

\subparagraph{Our results.}

In this paper, we consider the pathwidth of 2-layer $k$-planar graphs and show that the upper bound $k+1$ of~\cite{TheUpperBoundGD2020,TheUpperBoundComputJ2024} is sharp.
To this end, we give a 2-layer $k$-planar graph with pathwidth exactly $k+1$ for every $k \geq 0$, improving their lower bound $(k+3)/2$ of~\cite{TheUpperBoundGD2020,TheUpperBoundComputJ2024}.

\subparagraph{Related results.}

An outer $k$-planar drawing is a drawing such that the vertices are placed on a circle, the edges are straight-line segments, and every edge involves at most $k$ crossings.
\emph{Outer $k$-planar graphs}, the graphs that admit an outer $k$-planar drawing, are known to have treewidth at most $1.5k + 2$~\cite{DBLP:conf/gd/FirmanGKO024}, for which Pyzik~\cite{DBLP:conf/gd/Pyzik25} gave a lower bound of $1.5k + 0.5$.

\section{Preliminaries}

In this section, we give formal definitions for 2-layer $k$-planar graphs, pathwidth, and node searching number, which we use to give the lower bound of pathwidth, and some useful lemmas.
We follow the standard notations and terminology in graph theory (see, for example, \cite{diestel2025graph}).
For an integer $n \geq 1$, let $[n]$ denote the set $\{1, 2, \dots, n\}$.
For integers $n_{\ell} \leq n_r$, let $[n_{\ell}, n_r]$ denote the set $\{n_{\ell}, n_{\ell} + 1, \dots, n_r\}$.

\subparagraph{2-layer $k$-planar graphs.}

Since the proofs in this paper do not require the use of actual embeddings, we define 2-layer $k$-planar graphs combinatorially.
It can be easily confirmed that the definition below is equivalent to the (topological) one used in~\cite{TheUpperBoundGD2020,TheUpperBoundComputJ2024}.

Let $G = (X \cup Y, E)$ be a bipartite graph with bipartition $(X, Y)$.
Let $n_X = |X|$ and $n_Y = |Y|$.
Let $\pi_X \colon X \to [n_X]$, $\pi_Y \colon Y \to [n_Y]$ be bijections.
\emph{A 2-layer drawing} of $G$ is a pair of bijections $\pi = (\pi_X, \pi_Y)$.
On a 2-layer drawing $\pi$, an edge $\{x_1, y_1\} \in E$ \emph{crosses} an edge $\{x_2, y_2\} \in E$, where $x_1, x_2 \in X$ and $y_1, y_2 \in Y$, if and only if either one of $(\pi_X(x_1) < \pi_X(x_2)) \land (\pi_Y(y_1) > \pi_Y(y_2))$ or $(\pi_X(x_1) > \pi_X(x_2)) \land (\pi_Y(y_1) < \pi_Y(y_2))$ holds.
For an integer $k \geq 0$, a 2-layer drawing $\pi$ is \emph{a 2-layer $k$-planar drawing} if every edge in $E$ is involved in at most $k$ crossings on $\pi$.
The graph $G$ is \emph{a 2-layer $k$-planar graph} if it admits a 2-layer $k$-planar drawing.

\subparagraph{Pathwidth.}

Let $G = (V, E)$ be a graph.
A \emph{path decomposition} of $G$ is a pair of a path $P$ and a family of subsets $\mathcal{V} = (V_p)_{p \in V(P)}$ such that:
\begin{itemize}
    \item $V = \bigcup_p V_p$;
    \item for every edge $\{u, v\} \in E$, there exists $V_p \in \mathcal{V}$ such that $u, v \in V_p$; and
    \item for every vertex $v \in V$, the subgraph of $P$ induced by $\{p \in V(P) \mid v \in V_p\}$ is connected.
\end{itemize}
The \emph{width} of a path decomposition $(P, \mathcal{V})$ is defined as $\max_{p \in V(P)} |V_p| - 1$.
The \emph{pathwidth} of a graph $G$, denoted by $\pw(G)$, is the minimum width of a path decomposition of $G$.

\subparagraph{Node searching number.}

\emph{Node searching} is a one-player game played on a graph.
The edges are initially \emph{contaminated} and the goal is to \emph{clean} all the edges.
The possible moves in a turn are either placing or removing a \emph{guard} on a vertex.
A vertex is \emph{guarded} when a guard is placed on the vertex.
An edge becomes clean if the endpoints are both guarded.
An edge becomes contaminated if it shares a non-guarded endpoint with a contaminated edge.
We call this \emph{recontamination}.
After each turn, recontamination spreads as far as possible via non-guarded vertices.
A \emph{search strategy} is a sequence of moves from the initial configuration, where the edges are all contaminated and there is no guard, to a configuration where all the edges are clean.
The \emph{cost} of a search strategy is the maximum number of guards placed at the same time in the strategy.
For a graph $G$, the \emph{node searching number} of $G$, denoted by $\ns(G)$, is the minimum cost of a search strategy on $G$.
Kirousis and Papadimitriou~\cite{DBLP:journals/dm/KirousisP85} showed that $\ns(G)$ is identical to interval thickness, which is identical to pathwidth plus one~\cite[Theorem 29]{DBLP:journals/tcs/Bodlaender98}.
\begin{lemma}[\cite{DBLP:journals/tcs/Bodlaender98,DBLP:journals/dm/KirousisP85}]\label{lem:ns-pw}
    For every graph $G$, $\ns(G) = \pw(G) + 1$ holds.
\end{lemma}
It is known that allowing recontamination to happen does not help to decrease the number of guards required~\cite{DBLP:journals/tcs/KirousisP86,DBLP:journals/jacm/LaPaugh93}.
This allows us to consider only search strategies with no recontamination.
\begin{lemma}[\cite{DBLP:journals/tcs/KirousisP86,DBLP:journals/jacm/LaPaugh93}]\label{lem:no-recontamination}
    For every graph $G$, there exists a search strategy on $G$ with cost $\ns(G)$ that does not cause recontamination.
\end{lemma}

\section{Lower Bound}

In this section, we show our main result.

\begin{theorem}\label{thm:main}
    For every $k \geq 0$, there exists a 2-layer $k$-planar graph with pathwidth $k + 1$.
\end{theorem}

For $k = 0$, the path consisting of two vertices clearly satisfies the conditions.
Hence, in the following we show the claim for the case where $k \geq 1$.
To this end, we first construct a grid-like graph $G_k$ with pathwidth $k + 1$.
We then split its vertices so that the resulting graph $W_k$ admits a 2-layer $k$-planar drawing, preserving its pathwidth.

For an integer $k \geq 1$, let $G_k$ be a graph with vertex set $V_k = [k + 2] \times [3k + 6]$ and edge set $E_k = \{\{(r, c), (r + 1, c)\} \mid  r \in [k + 1], c \in [3k + 6]\} \cup \{\{(r, c), (r, c + 1)\} \mid r \in [k], c \in [3k + 5]\}$; see \cref{fig:G_k}.
For $r \in [k+2]$, we call the set of vertices $\{(r, c) \mid c \in [3k + 6]\}$ \emph{row} $r$.
Similarly, for $c \in [3k + 6]$, we call the set of vertices $\{(r, c) \mid r \in [k + 2]\}$ \emph{column} $c$.
We call an edge \emph{a row edge} (\emph{a column edge}) if the endpoints are in the same row (column).

When analyzing a search strategy on $G_k$, we say that a row (column) is \emph{clean} if all the row (column) edges on the row (column) are clean.

\begin{figure}[h]
    \centering
    \includegraphics[page=1]{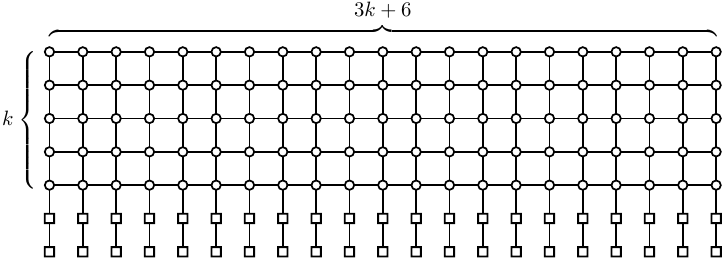}
    \caption{An illustration of $G_k$, which has $k + 2$ rows and $3k + 6$ columns.}
    \label{fig:G_k}
\end{figure}

\begin{lemma}\label{lem:pathwidth}
    For every $k \geq 1$, $\pw(G_k) = k + 1$.
\end{lemma}

\begin{proof}
    It is not difficult to see that $G_k$ is a minor of the $(k + 1) \times (6k + 12)$ grid, which has pathwidth $k+1$~\cite[Theorem 4.1]{DBLP:journals/dam/EllisW08}.
    As pathwidth is minor-monotone~\cite[Lemma 16]{DBLP:journals/tcs/Bodlaender98}, $\pw(G_k) \leq k + 1$ follows.
    Note also that $\pw(G_k) \geq k$ is clear as $G_k$ contains as a subgraph the $k \times k$ grid, which has pathwidth $k$~\cite[Theorem 4.1]{DBLP:journals/dam/EllisW08}.
    Hence, we have $\pw(G_k) \in \{k, k + 1\}$ and it suffices to disprove $\pw(G_k) = k$.

    Assume for contradiction that $\pw(G_k) = k$.
    Then, there exists a search strategy $S$ with cost $k + 1$.
    We further assume that $S$ does not cause recontamination by \cref{lem:no-recontamination}, and employ the following observation.
    This is almost the same as \cite[Observation 3.2]{DBLP:journals/dam/EllisW08}.

    \begin{observation}\label{obs:partly-cleaned}
        If a row $r \in [k]$ has both contaminated and clean edges, then there must be at least one guard on the row $r$.
        This property also holds for every column $c \in [3k+6]$.
    \end{observation}

    First, observe that in the search strategy $S$, none of the rows $1, \dots, k$ can have become clean unless at least $2k + 5$ columns are already clean.
    Otherwise, there are at most $2k + 4$ clean columns and at most $k + 1$ columns with a guard on them, leaving at least one column that is neither clean nor has a guard on it.
    By \cref{obs:partly-cleaned}, this column has only contaminated edges and would hence recontaminate the row, contradicting the assumption that $S$ causes no recontamination.

    Next, observe that once $k + 2$ columns become clean, each of the rows $1, \dots, k$ must contain a clean edge in $S$.
    Otherwise, there exists a row with all its edges being contaminated.
    Hence, to prevent recontamination, we must place guards on at least $k + 2$ intersections with the clean columns, which contradicts the cost of $k + 1$.

    Combining the above two observations and \cref{obs:partly-cleaned}, if the number of clean columns is in $[k + 2, 2k + 4]$, there must be at least one guard on each of the rows $1, \dots, k$.
    Let $c_{i}$ denote the $i$-th column to become clean.
    Note that two columns cannot become clean in the same turn and hence this is uniquely determined.
    By \cref{obs:partly-cleaned}, when $c_{k + 2}$ becomes clean, at least one of the $k + 2$ columns, $c_{k + 3}, c_{k + 4}, \dots, c_{2k + 4}$, has no clean edge.
    Let $c$ be such a column.
    Consider the turn when the edge $\{(k + 1, c), (k + 2, c)\}$ becomes clean.
    Right after this turn, there are still at most $2k + 4$ clean columns, and hence at least $k$ guards are placed on the $k$ other rows.
    This implies that there are at least $k + 2$ guards placed, which contradicts the cost of $k + 1$.
\end{proof}

Next, for every $k \geq 1$, we construct a wall-like 2-layer $k$-planar graph $W_k$ containing $G_k$ as a minor.
Since pathwidth is minor-monotone~\cite[Lemma 16]{DBLP:journals/tcs/Bodlaender98}, $\pw(W_k) \geq k + 1$ follows from \cref{lem:pathwidth}.
Hence, showing the existence of such graphs is sufficient to prove \cref{thm:main}.
Note that $\pw(W_k) \leq k + 1$ follows when $W_k$ is a 2-layer $k$-planar graph.

We first initialize $W_k$ as a graph consisting only of $k$ rows with $\ell = 4k (3k + 6)$ vertices each; namely, we let $W_k$ be a graph with vertex set $\{ (r, c) \mid r \in [k], c \in [\ell]\}$ and edge set $\{ \{(r, c), (r, c + 1)\} \mid r \in [k], c \in [\ell - 1]\}$.
We then add edges corresponding to the column edges of $G_k$.
For every $c \in [3k + 6]$, we apply the following operations to $W_k$ (see \cref{fig:W_k}):
\begin{enumerate}
    \item for every $r \in [k-1]$, add an edge $\{(r, 4k(c-1) + 4r - 3), (r + 1, 4k(c-1) + 4r - 2)\}$,
    \item\label{item:add-two-vertices} add two vertices $(k + 1, 4kc - 2), (k + 2, 4kc - 1)$, and
    \item\label{item:add-two-edges} add two edges $\{(k, 4kc - 3), (k + 1, 4kc - 2)\}, \{(k + 1, 4kc - 2), (k + 2, 4kc - 1)\}$.
\end{enumerate}
We call a subgraph consisting of the two vertices and the two edges added in Steps \ref{item:add-two-vertices} and \ref{item:add-two-edges} for some $c$ \emph{a hair}.
We define rows and columns for $W_k$ similarly: we call $\{(r, c) \mid c \in [\ell]\}$ row $r$, and $\{(r, c) \mid r \in [k + 2], (r, c) \in V(W_k)\}$ column $c$.

\begin{figure}[h]
    \centering
    \includegraphics[page=2]{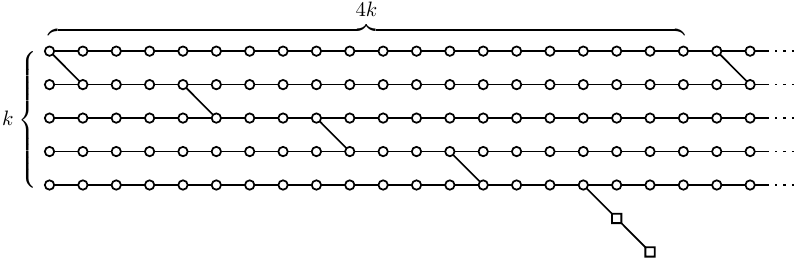}
    \caption{An illustration of $W_k$. The same pattern appears every $4k$ columns.}
    \label{fig:W_k}
\end{figure}

Now we show that the graph $W_k$ obtained in this manner satisfies the claimed conditions, which completes the proof of \cref{thm:main}.

\begin{lemma}
    For every $k \geq 1$, $W_k$ contains $G_k$ as a minor.
\end{lemma}

\begin{proof}
    For $r \in [k]$ and $c \in [3k + 6]$, let $S_{r, c} \subseteq V(W_k)$ be the vertex set $\{(r, c') \mid c' \in [4k(c-1) + 1, 4kc]\}$.
    By contracting $S_{r, c}$ into a single vertex $s_{r, c}$ for every $r, c$, we obtain a graph isomorphic to $G_k$. Note that $s_{r, c}$ corresponds to $(r, c) \in V(G_k)$.
\end{proof}

\begin{lemma}
    For every $k \geq 1$, $W_k$ is a 2-layer $k$-planar graph.
\end{lemma}

\begin{proof}
    Let $V_1 \subseteq V(W_k)$ be the vertex set $\{(r, c) \in V(W_k) \mid c \equiv 1 \pmod 2\}$ and $V_2 = V(W_k) \setminus V_1$.
    Observe that $(V_1, V_2)$ is a bipartition of $V(W_k)$.
    For $i \in \{1, 2\}$, let $\pi_i$ be the linear order of $V_i$ obtained by sorting $V_i$ in lexicographical order, where we define the key for a vertex $(r, c) \in V_i$ as $(c, r)$.
    We then claim that $\pi = (\pi_1, \pi_2)$ is a 2-layer $k$-planar drawing of $W_k$.
    Note that in a 2-layer drawing two edges do not cross more than once.
    Hence it suffices to show that every edge crosses at most $k$ other edges in $\pi$.

    First consider the subdrawing of $\pi$ induced by the row edges; see \cref{fig:rows}.
    In this subdrawing, a row edge $\{(r, c), (r, c + 1)\}$ crosses $k-r$ edges between columns $c-1, c$ and $r-1$ edges between columns $c+1, c+2$.
    Hence, this subdrawing is $(k-1)$-planar.

    \begin{figure}[h]
        \centering
        \includegraphics[page=3]{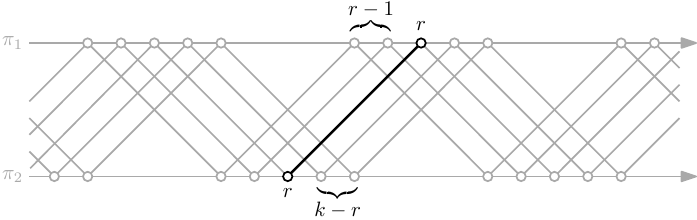}
        \caption{A part of the subdrawing of $\pi = (\pi_1, \pi_2)$ induced by the row edges.}
        \label{fig:rows}
    \end{figure}

    Next, we show that a non-row edge crosses at most $k$ other edges in $\pi$.
    There are two types of non-row edges: edges connecting two consecutive rows among rows $1, \dots, k$ (see \cref{fig:column-edge-normal}) and edges of the hairs attached to row $k$ (see \cref{fig:column-edge-hair}).
    Observe that non-row edges do not cross pairwise, since for every fourth column, only one of an edge of the first type or a single hair appears.
    Hence, as in \cref{fig:column-edges}, a non-row edge crosses at most $k$ edges regardless of its type.
    Note that we place a hair vertex immediately after the other vertices in the same column.

    \begin{figure}[h]
        \begin{subfigure}[b]{.49\textwidth}
            \centering
            \includegraphics[page=4]{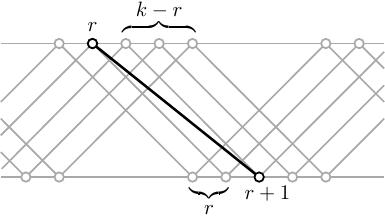}
            \caption{A non-row edge connecting two of rows $1, \dots, k$.}
            \label{fig:column-edge-normal}
        \end{subfigure}
        \hfill
        \begin{subfigure}[b]{.49\textwidth}
            \centering
            \includegraphics[page=5]{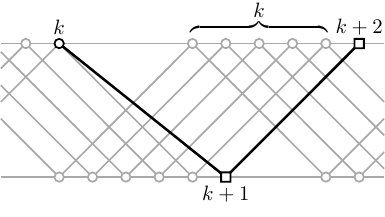}
            \caption{Two non-row edges forming a hair.}
            \label{fig:column-edge-hair}
        \end{subfigure}
        \caption{Two types of non-row edges.}
        \label{fig:column-edges}
    \end{figure}

    Lastly, we bound the number of crossings on a row edge in $\pi$.
    Consider a non-row edge of the first type.
    It is between columns $4t+1$ and $4t+2$ for some $t$.
    As in \cref{fig:column-edge-normal}, it crosses only row edges between two columns, $4t$ and $4t+1$, or $4t+2$ and $4t+3$.
    Next, consider a hair.
    It is attached to vertex $(k, 4t+1)$ for some $t$.
    Its first edge, namely the edge between rows $k$ and $k+1$, crosses only edges between columns $4t+2$ and $4t+3$.
    Similarly, its second edge crosses only edges between columns $4t+3$ and $4t+4$.
    Since these $t$'s are distinct, we can show that every row edge crosses at most one non-row edge in $\pi$ as follows.
    \begin{itemize}
        \item Consider a row edge between columns $4t + 0$ and $4t + 1$. Among non-row edges, this crosses only (if exists) the single edge of the first type between columns $4t + 1$ and $4t + 2$.
        \item Consider a row edge between columns $4t + 1$ and $4t + 2$. This does not cross any non-row edges.
        \item Consider a row edge between columns $4t + 2$ and $4t + 3$. Among non-row edges, this crosses only the single edge of the first type between columns $4t + 1$ and $4t + 2$ or otherwise the first edge of the hair attached to vertex $(k, 4t + 1)$.
        \item Consider a row edge between columns $4t + 3$ and $4t + 4$. Among non-row edges, this crosses only (if exists) the second edge of the hair attached to vertex $(k, 4t + 1)$.
    \end{itemize}
    Recall that every row edge crosses at most $k-1$ row edges.
    Hence, every row edge crosses at most $k$ other edges in $\pi$.
\end{proof}

\section{Conclusions}

In this paper, we gave a family of 2-layer $k$-planar graphs to show that the upper bound $k+1$ on the pathwidth of 2-layer $k$-planar graphs is sharp.

For future work, filling the gap in treewidth bounds for outer $k$-planar graphs ($1.5k+2$~\cite{DBLP:conf/gd/FirmanGKO024} and $1.5k+0.5$~\cite{DBLP:conf/gd/Pyzik25}) would be an interesting open problem to be revisited.
The lower bound of $1.5k+0.5$~\cite{DBLP:conf/gd/Pyzik25} shares a base idea with ours.
It is achieved by nicely arranging the vertices of a graph that contains two large grid graphs.
Hence, the splitting idea used in \cref{thm:main}, splitting vertices sufficiently to untangle parts with many crossings, might be helpful to improve the lower bound.

\bibliography{ref}

@article{TheUpperBoundComputJ2024,
  author    = {Patrizio Angelini and Giordano Da Lozzo and Henry F{\"{o}}rster and Thomas Schneck},
  title     = {2-Layer \emph{k}-Planar Graphs Density, Crossing Lemma, Relationships And Pathwidth},
  journal   = {Comput. J.},
  volume    = {67},
  number    = {3},
  pages     = {1005--1016},
  year      = {2024},
  doi       = {10.1093/COMJNL/BXAD038},
  bibsource = {dblp computer science bibliography, https://dblp.org}
}

@inproceedings{TheUpperBoundGD2020,
  author    = {Patrizio Angelini and Giordano Da Lozzo and Henry F{\"{o}}rster and Thomas Schneck},
  editor    = {David Auber and Pavel Valtr},
  title     = {2-Layer k-Planar Graphs - Density, Crossing Lemma, Relationships, and Pathwidth},
  booktitle = {Graph Drawing and Network Visualization - 28th International Symposium, {GD} 2020, Vancouver, BC, Canada, September 16-18, 2020, Revised Selected Papers},
  series    = {Lecture Notes in Computer Science},
  volume    = {12590},
  pages     = {403--419},
  publisher = {Springer},
  year      = {2020},
  doi       = {10.1007/978-3-030-68766-3_32},
  bibsource = {dblp computer science bibliography, https://dblp.org}
}

@inproceedings{DBLP:conf/gd/FirmanGKO024,
  author    = {Oksana Firman and Grzegorz Gutowski and Myroslav Kryven and Yuto Okada and Alexander Wolff},
  editor    = {Stefan Felsner and Karsten Klein},
  title     = {Bounding the Treewidth of Outer k-Planar Graphs via Triangulations},
  booktitle = {32nd International Symposium on Graph Drawing and Network Visualization, {GD} 2024, September 18-20, 2024, Vienna, Austria},
  series    = {LIPIcs},
  volume    = {320},
  pages     = {14:1--14:17},
  publisher = {Schloss Dagstuhl - Leibniz-Zentrum f{\"{u}}r Informatik},
  year      = {2024},
  doi       = {10.4230/LIPICS.GD.2024.14},
  bibsource = {dblp computer science bibliography, https://dblp.org}
}

@inproceedings{DBLP:conf/gd/Pyzik25,
  author    = {Rafal Pyzik},
  editor    = {Vida Dujmovic and
               Fabrizio Montecchiani},
  title     = {Treewidth of Outer k-Planar Graphs},
  booktitle = {33rd International Symposium on Graph Drawing and Network Visualization,
               {GD} 2025, Norrk{\"{o}}ping, Sweden, September 24-26, 2025},
  series    = {LIPIcs},
  volume    = {357},
  pages     = {28:1--28:16},
  publisher = {Schloss Dagstuhl - Leibniz-Zentrum f{\"{u}}r Informatik},
  year      = {2025},
  url       = {https://doi.org/10.4230/LIPIcs.GD.2025.28},
  doi       = {10.4230/LIPICS.GD.2025.28},
  bibsource = {dblp computer science bibliography, https://dblp.org}
}

@article{DBLP:journals/jgaa/Wood23,
  author    = {David R. Wood},
  title     = {2-Layer Graph Drawings with Bounded Pathwidth},
  journal   = {J. Graph Algorithms Appl.},
  volume    = {27},
  number    = {9},
  pages     = {843--851},
  year      = {2023},
  doi       = {10.7155/JGAA.00647},
  bibsource = {dblp computer science bibliography, https://dblp.org}
}

@inproceedings{DBLP:conf/mfcs/BiedlC0MNR20,
  author    = {Therese Biedl and Steven Chaplick and Michael Kaufmann and Fabrizio Montecchiani and Martin N{\"{o}}llenburg and Chrysanthi N. Raftopoulou},
  editor    = {Javier Esparza and Daniel Kr{\'{a}}l'},
  title     = {Layered Fan-Planar Graph Drawings},
  booktitle = {45th International Symposium on Mathematical Foundations of Computer Science, {MFCS} 2020, August 24-28, 2020, Prague, Czech Republic},
  series    = {LIPIcs},
  volume    = {170},
  pages     = {14:1--14:13},
  publisher = {Schloss Dagstuhl - Leibniz-Zentrum f{\"{u}}r Informatik},
  year      = {2020},
  doi       = {10.4230/LIPICS.MFCS.2020.14},
  bibsource = {dblp computer science bibliography, https://dblp.org}
}

@inproceedings{kobayashi_et_al:LIPIcs.SoCG.2025.65,
  author    = {Kobayashi, Yasuaki and Okada, Yuto and Wolff, Alexander},
  title     = {{Recognizing 2-Layer and Outer k-Planar Graphs}},
  booktitle = {41st International Symposium on Computational Geometry (SoCG 2025)},
  pages     = {65:1--65:16},
  series    = {Leibniz International Proceedings in Informatics (LIPIcs)},
  isbn      = {978-3-95977-370-6},
  issn      = {1868-8969},
  year      = {2025},
  volume    = {332},
  editor    = {Aichholzer, Oswin and Wang, Haitao},
  publisher = {Schloss Dagstuhl -- Leibniz-Zentrum f{\"u}r Informatik},
  address   = {Dagstuhl, Germany},
  urn       = {urn:nbn:de:0030-drops-232170},
  doi       = {10.4230/LIPIcs.SoCG.2025.65},
  bibsource = {dblp computer science bibliography, https://dblp.org}
}

@article{DBLP:journals/dam/EllisW08,
  author    = {John Ellis and Robert Warren},
  title     = {Lower bounds on the pathwidth of some grid-like graphs},
  journal   = {Discret. Appl. Math.},
  volume    = {156},
  number    = {5},
  pages     = {545--555},
  year      = {2008},
  url       = {https://doi.org/10.1016/j.dam.2007.02.006},
  doi       = {10.1016/j.dam.2007.02.006},
  bibsource = {dblp computer science bibliography, https://dblp.org}
}

@article{DBLP:journals/dm/KirousisP85,
  author    = {Lefteris M. Kirousis and Christos H. Papadimitriou},
  title     = {Interval graphs and searching},
  journal   = {Discret. Math.},
  volume    = {55},
  number    = {2},
  pages     = {181--184},
  year      = {1985},
  url       = {https://doi.org/10.1016/0012-365X(85)90046-9},
  doi       = {10.1016/0012-365X(85)90046-9},
  bibsource = {dblp computer science bibliography, https://dblp.org}
}

@article{DBLP:journals/jacm/LaPaugh93,
  author    = {Andrea S. LaPaugh},
  title     = {Recontamination Does Not Help to Search a Graph},
  journal   = {J. {ACM}},
  volume    = {40},
  number    = {2},
  pages     = {224--245},
  year      = {1993},
  url       = {https://doi.org/10.1145/151261.151263},
  doi       = {10.1145/151261.151263},
  bibsource = {dblp computer science bibliography, https://dblp.org}
}

@article{DBLP:journals/tsmc/SugiyamaTT81,
  author    = {Kozo Sugiyama and Shojiro Tagawa and Mitsuhiko Toda},
  title     = {Methods for Visual Understanding of Hierarchical System Structures},
  journal   = {{IEEE} Trans. Syst. Man Cybern.},
  volume    = {11},
  number    = {2},
  pages     = {109--125},
  year      = {1981},
  url       = {https://doi.org/10.1109/TSMC.1981.4308636},
  doi       = {10.1109/TSMC.1981.4308636},
  bibsource = {dblp computer science bibliography, https://dblp.org}
}

@article{DBLP:journals/algorithmica/DujmovicFKLMNRRWW08,
  author    = {Vida Dujmovic and Michael R. Fellows and Matthew Kitching and Giuseppe Liotta and Catherine McCartin and Naomi Nishimura and Prabhakar Ragde and Frances A. Rosamond and Sue Whitesides and David R. Wood},
  title     = {On the Parameterized Complexity of Layered Graph Drawing},
  journal   = {Algorithmica},
  volume    = {52},
  number    = {2},
  pages     = {267--292},
  year      = {2008},
  url       = {https://doi.org/10.1007/s00453-007-9151-1},
  doi       = {10.1007/s00453-007-9151-1},
  bibsource = {dblp computer science bibliography, https://dblp.org}
}

@article{DBLP:journals/jgaa/BinucciCDGKKMT17,
  author    = {Carla Binucci and Markus Chimani and Walter Didimo and Martin Gronemann and Karsten Klein and Jan Kratochv{\'{\i}}l and Fabrizio Montecchiani and Ioannis G. Tollis},
  title     = {Algorithms and Characterizations for 2-Layer Fan-planarity: From Caterpillar to Stegosaurus},
  journal   = {J. Graph Algorithms Appl.},
  volume    = {21},
  number    = {1},
  pages     = {81--102},
  year      = {2017},
  url       = {https://doi.org/10.7155/jgaa.00398},
  doi       = {10.7155/jgaa.00398},
  bibsource = {dblp computer science bibliography, https://dblp.org}
}

@article{doi:10.1137/0604033,
  author  = {Garey, M. R. and Johnson, D. S.},
  title   = {Crossing Number is {NP}-Complete},
  journal = {SIAM Journal on Algebraic Discrete Methods},
  volume  = {4},
  number  = {3},
  pages   = {312-316},
  year    = {1983},
  doi     = {10.1137/0604033},
  url     = {https://doi.org/10.1137/0604033}
}

@article{DBLP:journals/siamcomp/HeathR92,
  author    = {Lenwood S. Heath and Arnold L. Rosenberg},
  title     = {Laying out Graphs Using Queues},
  journal   = {{SIAM} J. Comput.},
  volume    = {21},
  number    = {5},
  pages     = {927--958},
  year      = {1992},
  url       = {https://doi.org/10.1137/0221055},
  doi       = {10.1137/0221055},
  bibsource = {dblp computer science bibliography, https://dblp.org}
}

@book{diestel2025graph,
  title     = {Graph Theory},
  edition   = {6th},
  author    = {Reinhard Diestel},
  series    = {Graduate Texts in Mathematics},
  volume    = {173},
  year      = {2025},
  publisher = {Springer Berlin},
  doi       = {10.1007/978-3-662-70107-2}
}

@article{DBLP:journals/tcs/KirousisP86,
  author    = {Lefteris M. Kirousis and Christos H. Papadimitriou},
  title     = {Searching and Pebbling},
  journal   = {Theor. Comput. Sci.},
  volume    = {47},
  number    = {3},
  pages     = {205--218},
  year      = {1986},
  url       = {https://doi.org/10.1016/0304-3975(86)90146-5},
  doi       = {10.1016/0304-3975(86)90146-5},
  bibsource = {dblp computer science bibliography, https://dblp.org}
}

@misc{2508.17349,
  title      = {2-Layer Fan-Planarity in Polynomial Time},
  author     = {Yasuaki Kobayashi and Yuto Okada},
  year       = {2025},
  eprint     = {2508.17349},
  eprinttype = {arXiv},
  url        = {https://arxiv.org/abs/2508.17349}
}

@article{DBLP:journals/tcs/Bodlaender98,
  author    = {Hans L. Bodlaender},
  title     = {A Partial \emph{k}-Arboretum of Graphs with Bounded Treewidth},
  journal   = {Theor. Comput. Sci.},
  volume    = {209},
  number    = {1-2},
  pages     = {1--45},
  year      = {1998},
  url       = {https://doi.org/10.1016/S0304-3975(97)00228-4},
  doi       = {10.1016/S0304-3975(97)00228-4},
  bibsource = {dblp computer science bibliography, https://dblp.org}
}

\end{document}